\documentclass[12pt, a4paper]{amsart}

\usepackage{amsmath,amssymb,amsfonts}
\usepackage{bbold}
\usepackage{epic,eepic}
\usepackage{enumerate}

\usepackage{amsthm}
\usepackage[all]{xy}
\usepackage{caption}
\usepackage{stmaryrd} 
\usepackage[dvipdfmx]{graphicx,color} 
\usepackage{tikz}
\usetikzlibrary{trees}

\theoremstyle{plain}
\newtheorem{thm}{Theorem}[section]
\newtheorem{prop}[thm]{Proposition}
\newtheorem{lem}[thm]{Lemma}

\theoremstyle{definition}
\newtheorem{defn}[thm]{Definition}

\numberwithin{figure}{section}

\numberwithin{table}{section}

\DeclareMathOperator*{\esssup}{ess\,sup}

\newcommand{\lspace} {
  \vspace{0.8\baselineskip}
}

\newcommand{\abs}[1]{
  \lvert  #1 \rvert
}

\newcommand{\norm}[1]{
  \|  #1 \|
}

\newcommand{\Nat}{\mathrm{Nat}}

\newdir{ >}{{}*!/-5pt/@{>}}

\definecolor{arrowred}{rgb}{0,0,0} 
\newcommand{\newword}[1]{\textbf{\textit{#1}}}

\numberwithin{equation}{section}

\newcommand{\ProbX}{\bar{X}}
\newcommand{\ProbY}{\bar{Y}}
\newcommand{\ProbZ}{\bar{Z}}
\newcommand{\ProbNX}{\bar{X}}
\newcommand{\ProbNY}{\bar{Y}}
\newcommand{\ProbNZ}{\bar{Z}}
\newcommand{\Prob}{\mathbf{Prob}}

\newcommand{\Set}{\mathbf{Set}}

\mathchardef\mhyphen="2D  

\title [Monetary Value Measures]{Monetary Value Measures in a Category of Probability Spaces}

\thanks{This work was supported by JSPS KAKENHI Grant Number 26330026.}

\author[T. Adachi and Y. Ryu]{Takanori Adachi and Yoshihiro Ryu}
\address{Department of Mathematical Sciences,
         Ritsumeikan University,
         1-1-1 Nojihigashi, Kusatsu, Shiga, 525-8577 Japan}
\email{Takanori Adachi <taka.adachi@gmail.com>}
\email{Yoshihiro Ryu <iti2san@gmail.com>}



\begin{document}

\maketitle



In this paper,
we generalize the notion of monetary value measures
developed in 
\cite{adachi_2014crm}
by extending their base category from
the category
$\chi$
to
 the category
$\Prob$
introduced in
\cite{AR_2016_1}.

For those who are not familiar with financial risk management and/or monetary value measures,
please refer to Section 2 of 
\cite{adachi_2014crm}.


\section{A Category of Probability Spaces}
\label{sec:cat_prob}

In this section we overview 
a theory of a category of probability spaces.
Please refer to
\cite{AR_2016_1}
for the full discussions and proofs 
about the contents of this section.

\lspace

Let
$
\bar{X} :=
(X, \Sigma_X, \mathbb{P}_X)
$,
$
\bar{Y} :=
(Y, \Sigma_Y, \mathbb{P}_Y)
$
and
$
\bar{Z} :=
(Z, \Sigma_Z, \mathbb{P}_Z)
$
be
probability spaces.

\begin{defn}{[Category $\Prob$]}
\label{defn:cat_prob}
A category
$
\Prob
$
is the category whose objects are
all probability spaces and 
the set of arrows between them are defined by
\begin{align*}
\Prob(
  \ProbX,
&
  \ProbY
)
  :=
\{
f^-
  \mid
f :
  \ProbY
\to
  \ProbX
\\&
\textrm{ is a measurable function satisfying }
  \mathbb{P}_Y \circ f^{-1} \ll \mathbb{P}_X
\},
\end{align*}
where
$f^-$
is a symbol
corresponding uniquely to
a measurable function $f$.


\end{defn}

We fix the state space be
a measurable space
$
(
 \mathbb{R}, 
 \mathcal{B}(\mathbb{R})
)
$
for a simplicity.
$
\mathcal{L}^{\infty}(\ProbNX)
$
is a vector space consisting of 
$\mathbb{R}$-valued random variables
$v$
such that
$
\mathbb{P}_X\mhyphen\esssup_{x \in X}
\abs{v(x)}
  <
\infty
$,
while
$
\mathcal{L}^1(\ProbNX)
$
is a vector space consisting of 
$\mathbb{R}$-valued random variables
$v$
such that
$
\int_X
  \abs{v}
\, d \mathbb{P}_X
$
has a finite value.
For
two random variables
$
u_1
$
and
$
u_2
$,
we write
$
u_1
  \sim_{\mathbb{P}_X}
u_2
$
when
$
\mathbb{P}_X(
  u_1 \ne u_2
)
  =
0
$,
and write
$
   u_1
       \lesssim_{\mathbb{P}_X}
   u_2 
$
when
$
\mathbb{P}_X(
  u_1 > u_2
)
  =
0
$.
Note that
$
   u_1
       \lesssim_{\mathbb{P}_X}
   u_2 
$
and
$
   u_2
       \lesssim_{\mathbb{P}_X}
   u_1 
$
iff
$
   u_1
       \sim_{\mathbb{P}_X}
   u_2 
$.
$
    L^{\infty}(\ProbNX)
$
and
$
    L^1(\ProbNX)
$
are quotient spaces
$
  \mathcal{L}^{\infty}(\ProbNX)
/
  \sim_{\mathbb{P}_X}
$
and
$
  \mathcal{L}^1(\ProbNX)
/
  \sim_{\mathbb{P}_X}
$,
respectively.

\begin{defn}{[Functor $\mathbf{L}$]}
\label{defn:fun_L}
A functor
$
\mathbf{L} : \Prob \to \Set
$
is defined by:
\begin{equation*}
\xymatrix@C=15 pt@R=30 pt{
    X
  &
    \ProbX
     \ar @{->}_{f^-} [d]
     \ar @{|->}^{\mathbf{L}} [rr]
  &&
    \mathbf{L}\ProbX
     \ar @{}^-{:=} @<-6pt> [r]
     \ar @{->}^{
       \mathbf{L} f^-
     } [d]
  &
    L^{\infty}(\ProbNX)
     \ar @{}^-{\ni} @<-6pt> [r]
  &
    [u]_{\sim_{\mathbb{P}_X}}
     \ar @{|->}^{
       \mathbf{L} f^-
     } @<-9pt> [d]
\\
    Y
      \ar @{->}^f [u]
  &
    \ProbY
     \ar @{|->}^{\mathbf{L}} [rr]
  &&
    \mathbf{L}\ProbY
     \ar @{}^-{:=} @<-6pt> [r]
  &
    L^{\infty}(\ProbNY)
     \ar @{}^-{\ni} @<-6pt> [r]
  &
    [u \circ f]_{\sim_{\mathbb{P}_Y}}
}
\end{equation*}
\end{defn}

\begin{thm}
\label{thm:CEf}
Let
$
f^-
$
be an arrow in
$
\Prob(
  \ProbX,
  \ProbY
)
$.
Then, for any
$
v \in
\mathcal{L}^{1}(\ProbNY)
$
there exists a
$
u \in
\mathcal{L}^{1}(\ProbNX)
$
such that
for every
$
A \in \Sigma_X
$
\begin{equation}
\label{eq:RN}
\int_A u \, d \mathbb{P}_X
  =
\int_{f^{-1}(A)} v \, d \mathbb{P}_Y.
\end{equation}
Moreover,
$u$
is determined uniquely up to 
$\mathbb{P}_X$-null sets.
In other words,
if there are two
$
u_1, u_2 \in
\mathcal{L}^{1}(\ProbNX)
$
both satisfying
(\ref{eq:RN}),
then
$
u_1
  \sim_{\mathbb{P}_X}
u_2
$.

We write a version of this 
$u$
by
$
E^{f^-}(v)
$,
and call it a
\newword{conditional expectation of $v$ along $f^-$}.
Therefore,
\begin{equation}
\label{eq:RN_E}
\int_A 
  E^{f^-}(v)
\, d \mathbb{P}_X
  =
\int_{f^{-1}(A)} v \, d \mathbb{P}_Y.
\end{equation}
\end{thm}

\begin{prop}
\label{prop:fun_E}
Let
$f^-$
and
$g^-$
be arrows in 
$\Prob$
like:
\begin{equation*}
\xymatrix{
  \ProbX
       \ar @{->}^{f^-} [r]
 &
  \ProbY
       \ar @{->}^{g^-} [r]
 &
  \ProbZ
}.
\end{equation*}
\begin{enumerate}
\item
For
$
u \in 
\mathcal{L}^1(\ProbNX)
$,
$
E^{Id_X^-}(u)
  \sim_{\mathbb{P}_X}
u
$.

\item
For
$
v_1, v_2
  \in
\mathcal{L}^{1}(\ProbNY)
$,
$
v_1
  \sim_{\mathbb{P}_Y}
v_2
$
implies
$
E^{f^-}(v_1)
  \sim_{\mathbb{P}_X}
E^{f^-}(v_2)
$.

\item
For
$
w
  \in
\mathcal{L}^{1}(\ProbNZ)
$,
$
E^{f^-}(E^{g^-}(w))
  \sim_{\mathbb{P}_X}
E^{g^- \circ f^-}(w)
$.
\end{enumerate}
\end{prop}

\begin{defn}{[Functor $\mathcal{E}$]}
\label{defn:fun_E}
A functor
$
\mathcal{E} : \Prob^{op} \to \Set
$
is defined by:
\begin{equation*}
\xymatrix@C=15 pt@R=30 pt{
    X
  &
    \ProbX
     \ar @{->}_{f^-} [d]
     \ar @{|->}^{\mathcal{E}} [rr]
  &&
    \mathcal{E}\ProbX
     \ar @{}^-{:=} @<-6pt> [r]
  &
    L^{1}(\ProbNX)
     \ar @{}^-{\ni} @<-6pt> [r]
  &
    [E^{f^-}(v)]_{\sim_{\mathbb{P}_X}}
\\
    Y
      \ar @{->}^f [u]
  &
    \ProbY
     \ar @{|->}^{\mathcal{E}} [rr]
  &&
    \mathcal{E}\ProbY
     \ar @{}^-{:=} @<-6pt> [r]
     \ar @{->}_{
       \mathcal{E} f^-
     } [u]
  &
    L^{1}(\ProbNY)
     \ar @{}^-{\ni} @<-6pt> [r]
  &
    [v]_{\sim_{\mathbb{P}_Y}}
     \ar @{|->}_{
       \mathcal{E} f^-
     } @<9pt> [u]
}
\end{equation*}
We call
$\mathcal{E}$
a 
\newword{conditional expectation functor}.
\end{defn}

\begin{prop}
\label{prop:lin}
Let
$
f^-
  :
\ProbX
  \to
\ProbY
$
be a 
$\Prob$-arrow,
$
u, v
  \in
 \mathcal{L}^1(\ProbNY)
$
and
$
\alpha ,\beta
\in \mathbf{R}
$.
\begin{enumerate}
\item
    \noindent \textit{Linearity:}\;
$
E^{f^-} (
    \alpha u + \beta v
)
   \sim_{\mathbb{P}_X}
\alpha E^{f^-}(u) + \beta E^{f^-}(v) .
$

\item
    \noindent \textit{Positivity:}\;
$
E^{f^-} (v)
  \gtrsim_{\mathbb{P}_X}
0
$
if 
$
v
  \gtrsim_{\mathbb{P}_Y}
0
$.

\end{enumerate}

\end{prop}

\begin{thm}
\label{thm:measu}
Let 
$
f^-
  :
\ProbX
  \to
\ProbY
$
be a
$\Prob$-arrow
,
$
u 
  \in
\mathcal{L}^{1}(\ProbNY)
$
and
$
w 
  \in
\mathcal{L}^{\infty}(\ProbNX)
$.
Then
we have
\begin{equation}
E^{f^-}(
  (w \circ f) \cdot u
)
  \sim_{\mathbb{P}_X}
w 
  \cdot
E^{f^-}(u) .
\end{equation}

\end{thm}

\section{Monetary Value Measures}
\label{sec:value_measure}

A monetary value measure is defined as a presheaf 
on 
$\Prob$.

\begin{defn}{[Monetary Value Measures]}
\label{defn:value_measure}
A 
 \newword{
  monetary value measure
}
is a contravariant functor
\begin{equation*}
  \Phi : 
    \Prob^{op}
  \to
    \Set
\end{equation*}
defined by
\begin{equation*}
\xymatrix@C=15 pt@R=30 pt{
    X
  &
    \ProbX
     \ar @{->}_{f^-} [d]
     \ar @{|->}^{\Phi} [rr]
  &&
    \Phi\ProbX
     \ar @{}^-{:=} @<-6pt> [r]
  &
    L^{1}(\ProbNX)
     \ar @{}^-{\ni} @<-6pt> [r]
  &
    [\varphi^{f^-}(v)]_{\sim_{\mathbb{P}_X}}
\\
    Y
      \ar @{->}^f [u]
  &
    \ProbY
     \ar @{|->}^{\Phi} [rr]
  &&
    \Phi\ProbY
     \ar @{}^-{:=} @<-6pt> [r]
     \ar @{->}_{
       \Phi f^-
     } [u]
  &
    L^{1}(\ProbNY)
     \ar @{}^-{\ni} @<-6pt> [r]
  &
    [v]_{\sim_{\mathbb{P}_Y}}
     \ar @{|->}_{
       \Phi f^-
     } @<9pt> [u]
}
\end{equation*}
where
$\varphi^{f^-}$
satisfies
\begin{enumerate}
  \item
    \noindent \textit{Cash invariance:}\;
      $
      (\forall v \in \mathcal{L}^{\infty}(\ProbNY) )
      (\forall u \in \mathcal{L}^{\infty}(\ProbNX) )
      $

      $
  \varphi^{f^-}(v + u \circ f)
     \sim_{\mathbb{P}_X} 
  \varphi^{f^-}(v) + u
     $,

  \item
    \noindent \textit{Monotonicity:}\;
      $
      (\forall v_1 \in \mathcal{L}^{\infty}(\ProbNY) )
      (\forall v_2 \in \mathcal{L}^{\infty}(\ProbNY) )
      $

      $
   v_1
       \lesssim_{\mathbb{P}_Y}
   v_2 
       \; \Rightarrow \;
   \varphi^{f^-}(v_1)
       \lesssim_{\mathbb{P}_X}
   \varphi^{f^-}(v_2)
      $,

  \item
    \noindent \textit{Normalization:}\;
      $
      \varphi^{f^-}(0_Y)
         \sim_{\mathbb{P}_X} 
      0_X
      $
      if
      $f^-$
      is measure-preserving,

  \item
      $
         v \in \mathcal{L}^{\infty}(\ProbY)
      $
          implies
      $
         \varphi^{f^-}(v) \in \mathcal{L}^{\infty}(\ProbX)
      $
      if
      $f^-$
      is measure-preserving.

\end{enumerate}
We sometimes write
$
\Phi[\varphi^{\cdot}]
$
for
$
\Phi
$
for explicitly noting that
arrows mapped by 
$\Phi$
are determined by
$
\varphi^{\cdot}
$.
\end{defn}

At this point,
we do not require the monetary value measures to satisfy
familiar conditions
such as
concavity or positive homogeneity.
Instead of doing so, we want to see what kind of properties
are deduced from this minimal setting.

\lspace
The most crucial point of
Definition \ref{defn:value_measure}
is that 
$\varphi$
does not move only 
in the direction of time
but also moves over several absolutely continuous probability measures
\textit{internally}.
This means we have a possibility to develop risk measures including ambiguity
within this formulation.

Another key point of
Definition \ref{defn:value_measure}
is that 
$\varphi$
is a contravariant functor.
So, for any pair of arrows
$
  \xymatrix@C=15 pt@R=15 pt{
     \ProbX
        \ar @{->}_{f^-} [r]
   &
     \ProbY
        \ar @{->}_{g^-} [r]
   &
     \ProbZ
  }
$
in
$\Prob$,
we have
\begin{equation}
\label{eq:cv_functor}
\Phi Id_X^-
  =
Id_{L^1(\ProbNX)}
\;\; \textrm{and} \;\;
\Phi f^-
  \circ
\Phi g^-
  =
\Phi (g^- \circ f^-).
\end{equation}

As an example of monetary value measures,
we will introduce a notion of 
entropic value measures 
that depend on conditional expectations
$
E^{f^-}(v)
$
of 
$v$
along
$f^-$.

Before introducing entropic value measures,
we need the following lemma.

\begin{lem}
\label{lem:mpInfty}
Let
$
f^-
  :
\ProbX \to \ProbY
$
be a measure-preserving arrow in
$\Prob$.
Then, 
$
v
  \in \mathcal{L}^{\infty}(\ProbNY)
$
implies
$
E^{f^-}(v)
  \in \mathcal{L}^{\infty}(\ProbNX)
$.
\end{lem}
\begin{proof}
Since
$
v
  \in \mathcal{L}^{\infty}(\ProbNY)
$,
there exists a non-negative
$M \ge 0$
such that
$
-M
  \lesssim_{\mathbb{P}_Y}
v
  \lesssim_{\mathbb{P}_Y}
M
$.
Then,
by Proposition \ref{prop:lin},
\begin{equation*}
0
  \lesssim_{\mathbb{P}_X}
E^{f^-}(M - v)
  \sim_{\mathbb{P}_X}
E^{f^-}(M)
  -
E^{f^-}(v) .
\end{equation*}
On the other hand,
we have
\begin{equation*}
E^{f^-}(M)
  \sim_{\mathbb{P}_X}
M E^{f^-}(1_Y)
  \sim_{\mathbb{P}_X}
M
\end{equation*}
since 
$f$
is measure preserving.
Therefore, we obtain
$
E^{f^-}(v)
  \lesssim_{\mathbb{P}_X}
M
$.
Similarly, we have
$
-M
  \lesssim_{\mathbb{P}_X}
E^{f^-}(v)
$.
So we get
$
E^{f^-}(v)
  \in \mathcal{L}^{\infty}(\ProbNX)
$.
\end{proof}

\begin{prop}{[Entropic Value Measures]}
\label{prop:EVM}
Let
$
f^-
  :
\ProbX
  \to
\ProbY
$
be
 a 
$\Prob$-arrow,
and
$\lambda$
be a positive real number.
Define
a function
$
\varphi^{f^-}
  :
L^1(\ProbNY)
  \to
L^1(\ProbNX)
$
by 
\begin{equation}
\varphi^{f^-}(v)
  :=
\lambda^{-1}
\log 
  E^{f^-}
  (e^{\lambda v}),
\;
(
\forall v \in
L^1(\ProbNY)
).
\end{equation}
Then,
$
\Phi := \Phi[\varphi^{\cdot}]
$
is a monetary value measure.
We call this 
$\Phi$
an \newword{entropic value measure}.
\end{prop}
\begin{proof}
Let
$
  \xymatrix@C=15 pt@R=15 pt{
     \ProbX
        \ar @{->}^{f^-} [r]
   &
     \ProbY
        \ar @{->}^{g^-} [r]
   &
     \ProbZ
  }
$
be arrows in
$\Prob$.
In order to show that 
$\Phi$
becomes a contravariant functor,
we need to check three points:
$
\varphi^{Id_Z^-}(w)
  \sim_{\mathbb{P}_Z}
w
$,
$
\varphi^{f-}(
  \varphi^{g^-}(
     w
  )
)
  \sim_{\mathbb{P}_X}
\varphi^{g^- \circ f^-}(
 w
)
$,
and that
$
w_1
  \sim_{\mathbb{P}_Z}
w_2
$
implies
$
\varphi^{g^-}(
  w_1
)
  \sim_{\mathbb{P}_Y}
\varphi^{g^-}(
  w_2
)
$
for every
$
w, w_1, w_2
  \in
\mathcal{L}^1(\ProbZ)
$.
But, they are straightforward consequences of
Proposition \ref{prop:fun_E}.
So,
we forward to check if
$
\varphi^{f^-}
$
satisfies the four conditions of
Definition \ref{defn:value_measure}.

Firstly, we show that
$
  \varphi^{f^-}(v + u \circ f)
     \sim_{\mathbb{P}_X} 
  \varphi^{f^-}(v) + u
$
for
$
v \in \mathcal{L}^{\infty}(\ProbNY) 
$
and
$
u \in \mathcal{L}^{\infty}(\ProbNX) 
$.
But by
Theorem \ref{thm:measu},
we have
\begin{align*}
  \varphi^{f^-}(v + u \circ f)
&=
  \lambda^{-1} \log
  E^{f^-}(e^{\lambda (v + u \circ f)})
\\&=
  \lambda^{-1} \log
  E^{f^-}\big(
    e^{\lambda v}
  \cdot
    ((e^{\lambda u}) \circ f)
  \big)
\\&\sim_{\mathbb{P}_X} 
  \lambda^{-1} \log
  \Big(
    e^{\lambda u}
  \cdot
    E^{f^-}(
      e^{\lambda v}
    )
  \Big)
\\&=
  u
    +
  \varphi^{f^-}(v) .
\end{align*}

Secondly, we show that
$
   v_1
       \lesssim_{\mathbb{P}_Y}
   v_2 
$
implies
$
   \varphi^{f^-}(v_1)
       \lesssim_{\mathbb{P}_X}
   \varphi^{f^-}(v_2)
$
for
$
       v_1, v_2 \in \mathcal{L}^{\infty}(\ProbNY) 
$.
But this comes from
Proposition \ref{prop:lin}.

Thirdly, we show that
$
      \varphi^{f^-}(0_Y)
         \sim_{\mathbb{P}_X} 
      0_X
$
if
$f^-$
is measure-preserving.
But this is straightforward like the following:
\begin{equation*}
\varphi^{f^-}(0_Y)
  =
\lambda^{-1} \log
E^{f^-}(e^{\lambda 0_Y})
  =
\lambda^{-1} \log
E^{f^-}(1_Y)
  \sim_{\mathbb{P}_X} 
\lambda^{-1} \log
1_X
  =
0_X.
\end{equation*}

Lastly,
we need to show that
      $
         v \in \mathcal{L}^{\infty}(\ProbY)
      $
          implies
      $
         \varphi^{f^-}(v) \in \mathcal{L}^{\infty}(\ProbX)
      $
when
      $f^-$
      is measure-preserving.
But this comes from 
Lemma
\ref{lem:mpInfty}.
\qedhere
\end{proof}

Here are some properties of monetary value measures.

\begin{thm}
\label{thm:1}
Let 
$
\Phi = \Phi[\varphi^{\cdot}]
   : \Prob^{op} \to \Set
$
be a monetary value measure,
 and
$
  \xymatrix@C=15 pt@R=15 pt{
     \ProbX
        \ar @{->}^{f^-} [r]
   &
     \ProbY
        \ar @{->}^{g^-} [r]
   &
     \ProbZ
  }
$
be 
arrows in
$\Prob$.
\begin{enumerate}
\item
If 
$f^-$
is measure-preserving,
we have \;
$
   \Phi f^- 
\circ
   L f^-
=
   Id_{ L \ProbX }
$.

\item
    \noindent \textit{Idempotence:}\;
If 
$f^-$
is measure-preserving,
we have 

$
  \Phi f^-
\circ
  L f^-
\circ
  \Phi f^-
    = 
  \Phi f^-
$.

\item
    \noindent \textit{Local property:}\;
$
(\forall v_1 \in \mathcal{L}^{\infty}(\ProbY))
(\forall v_2 \in \mathcal{L}^{\infty}(\ProbY))
(\forall A \in \Sigma_X)
$
$
     \Phi f^- \big[
        1_{f^{-1}(A)} v_1
            +
        1_{f^{-1}(A^c)} v_2
     \big]_{\sim_{\mathbb{P}_Y}}
  =
     [ 1_{f^{-1}(A)} ]_{\sim_{\mathbb{P}_X}}
     \Phi f [v_1]_{\sim_{\mathbb{P}_Y}}
  +
     [ 1_{f^{-1}(A^c)} ]_{\sim_{\mathbb{P}_X}}
     \Phi f [v_2]_{\sim_{\mathbb{P}_Y}}
$.

\item
    \noindent \textit{Dynamic programming principle:}\;
If 
$g^-$
is measure-preserving,

$
  \varphi^{g^- \circ f^-} (w)
  =
  \varphi^{g^- \circ f^-} (
    \varphi^{g^-} (w)
        \circ
    g
  )
$
for
$
w
  \in
\mathcal{L}^{\infty}(\ProbZ)
$.

\item
    \noindent \textit{Time consistency:}\;
$
(\forall w_1 \in \mathcal{L}^{\infty}(\ProbZ))
(\forall w_2 \in \mathcal{L}^{\infty}(\ProbZ))
$

$
     \varphi^{g^-}(w_1)
   \lesssim_{\mathbb{P}_Y}
     \varphi^{g^-}(w_2)
\; \Rightarrow \;
     \varphi^{g^- \circ f^-}(w_1)
   \lesssim_{\mathbb{P}_X}
     \varphi^{g^- \circ f^-}(w_2)
$.

\end{enumerate}
\end{thm}
\begin{proof}
\begin{enumerate}
\item
For
$
  u \in \mathcal{L}^{\infty}(\ProbNX)
$,
$
     \Phi f^- (
        L f^- [
          u
        ]_{\sim_{\mathbb{P}_X}}
     )
=
     \big[
        \varphi^{f^-}(u \circ f)
     \big]_{\sim_{\mathbb{P}_X}}
$
But,
by 
cash invariance
and 
normalization,
we have
$
\varphi^{f^-}(u \circ f)
  =
\varphi^{f^-}(0_Y + (u \circ f))
  \sim_{\mathbb{P}_X}
\varphi^{f^-}(0_Y) + u
  \sim_{\mathbb{P}_X}
0_X + u
  =
u
$.

\item
Immediate by (1).

\item
First, we show that
for any
$
A \in \Sigma_X
$
and
$
v 
  \in
\mathcal{L}^{\infty}(\ProbNY)
$,
\begin{equation}
\label{eq:p51}
1_A \varphi^{f^-}(v)
    \sim_{\mathbb{P}_X}
1_A \varphi^{f^-}(1_{f^{-1}(A)} v) .
\end{equation}
Since
$
v \in 
\mathcal{L}^{\infty}(\ProbNY)
$,
for every
$
y \in Y
$
we have
$
\abs{v(y)}
   \lesssim_{\mathbb{P}_Y}
\norm{v}_{\mathcal{L}^{\infty}(\ProbNY)}
$.
Therefore,
\begin{equation*}
1_{f^{-1}(A)}
  v
-
1_{f^{-1}(A^c)}
  \norm{v}_{\mathcal{L}^{\infty}(\ProbNY)}
      \lesssim_{\mathbb{P}_Y}
1_{f^{-1}(A)}
  v
+
1_{f^{-1}(A^c)}
  v
      \lesssim_{\mathbb{P}_Y}
1_{f^{-1}(A)}
  v
+
1_{f^{-1}(A^c)}
  \norm{v}_{\mathcal{L}^{\infty}(\ProbNY)} .
\end{equation*}
Then
noting that
$
1_A \circ f = 1_{f^{-1}(A)}
$,
we have the following sequence of equations
by cash invariance and monotonicity. 
\begin{align*}
  \varphi^{f^-} (
     1_{f^{-1}(A)}
     v
  )
 -
  \norm{v}_{\mathcal{L}^{\infty}(\ProbNY)}
  1_{A^c}
& \sim_{\mathbb{P}_X}
  \varphi^{f^-} (
     1_{f^{-1}(A)}
     v
  -
     (\norm{v}_{\mathcal{L}^{\infty}(\ProbNY)}
     1_{A^c})
         \circ
     f
  )
\\&=
  \varphi^{f^-} (
     1_{f^{-1}(A)}
     v
  -
     1_{f^{-1}(A^c)}
     \norm{v}_{\mathcal{L}^{\infty}(\ProbNY)}
  )
\\& \lesssim_{\mathbb{P}_X}
  \varphi^{f^-} (v)
\\& \lesssim_{\mathbb{P}_X}
  \varphi^{f^-} (
     1_{f^{-1}(A)}
     v
  +
     1_{f^{-1}(A^c)}
     \norm{v}_{\mathcal{L}^{\infty}(\ProbNY)}
  )
\\&=
  \varphi^{f^-} (
     1_{f^{-1}(A)}
     v
  +
     (\norm{v}_{\mathcal{L}^{\infty}(\ProbNY)}
     1_{A^c})
         \circ
     f
  )
\\& \sim_{\mathbb{P}_X}
  \varphi^{f^-} (
     1_{f^{-1}(A)}
     v
  )
 +
  \norm{v}_{\mathcal{L}^{\infty}(\ProbNY)}
  1_{A^c} .
\end{align*}
Hence
\begin{equation*}
  \varphi^{f^-} (
     1_{f^{-1}(A)}
     v
  )
 -
  1_{A^c}
  \norm{v}_{\mathcal{L}^{\infty}(\ProbNY)}
         \lesssim_{\mathbb{P}_X}
  \varphi^{f^-} (v)
         \lesssim_{\mathbb{P}_X}
  \varphi^{f^-} (
     1_{f^{-1}(A)}
     v
  )
 +
  1_{A^c}
  \norm{v}_{\mathcal{L}^{\infty}(\ProbNY)} .
\end{equation*}
By multiplying
$
1_A
$,
we obtain
\begin{equation*}
  1_A
  \varphi^{f^-} (
     1_{f^{-1}(A)}
     v
  )
 -
         \lesssim_{\mathbb{P}_X}
  1_A
  \varphi^{f^-} (v)
         \lesssim_{\mathbb{P}_X}
  1_A
  \varphi^{f^-} (
     1_{f^{-1}(A)}
     v
  ).
\end{equation*}
Therefore, we get
 (\ref{eq:p51}).

Next by using 
 (\ref{eq:p51})
twice,
we have
\begin{align*}
&
  \varphi^{f^-} (
     1_{f^{-1}(A)}
       v_1
   +
     1_{f^{-1}(A^c)}
       v_2
  )
\\ &= 
  1_{f^{-1}(A)}
  \varphi^{f^-} (
     1_{f^{-1}(A)}
       v_1
   +
     1_{f^{-1}(A^c)}
       v_2
  )
+
  1_{f^{-1}(A^c)}
  \varphi^{f^-} (
     1_{f^{-1}(A)}
       v_1
   +
     1_{f^{-1}(A^c)}
       v_2
  )
\\ & \sim_{\mathbb{P}_X} 
  1_{f^{-1}(A)}
  \varphi^{f^-} (
    1_{f^{-1}(A)} (
     1_{f^{-1}(A)}
       v_1
   +
     1_{f^{-1}(A^c)}
       v_2
    )
  )
+
  1_{f^{-1}(A^c)}
  \varphi^{f^-} (
    1_{f^{-1}(A^c)} (
     1_{f^{-1}(A)}
       v_1
   +
     1_{f^{-1}(A^c)}
       v_2
    )
  )
\\&=
  1_{f^{-1}(A)}
  \varphi^{f^-} (
     1_{f^{-1}(A)}
       v_1
  )
+
  1_{f^{-1}(A^c)}
  \varphi^{f^-} (
     1_{f^{-1}(A^c)}
       v_2
  )
\\ & \sim_{\mathbb{P}_X} 
  1_{f^{-1}(A)}
  \varphi^{f^-} (
       v_1
  )
+
  1_{f^{-1}(A^c)}
  \varphi^{f^-} (
       v_2
  ).
\end{align*}

\item
By (2), we have
$
\varphi^{g^-} (
  \varphi^{g^-}(w)
    \circ
  g
)
   \sim_{\mathbb{P}_Y} 
\varphi^{g^-}(w)
$
for
$
w \in
  \mathcal{L}^{\infty}(\ProbNZ)
$.
So by
 (\ref{eq:cv_functor}),
\begin{align*}
     \varphi^{g^- \circ f^-}(w)
& \sim_{\mathbb{P}_X} 
   \varphi^{f^-}
   (\varphi^{g^-}(w))
\sim_{\mathbb{P}_X} 
   \varphi^{f^-}(
     \varphi^{g^-} (
       \varphi^{g^-}(w)
         \circ
       g
     )
   )
\\
&=
   (\varphi^{f^-}
      \circ
    \varphi^{g^-})(
       \varphi^{g^-}(w)
         \circ
       g
   )
\sim_{\mathbb{P}_X} 
   \varphi^{g^- \circ f^-}(
       \varphi^{g^-}(w)
         \circ
       g
   ) .
\end{align*}

\item
Assume
$
     \varphi^{g^-}(w_1)
   \lesssim_{\mathbb{P}_Y}
     \varphi^{g^-}(w_2)
$.
Then, by monotonicity and
 (\ref{eq:cv_functor}),
\begin{equation*}
     \varphi^{g^- \circ f^-}(w_1)
\sim_{\mathbb{P}_X} 
     \varphi^{f^-}(
       \varphi^{g^-}(w_1)
     )
\lesssim_{\mathbb{P}_X} 
     \varphi^{f^-}(
       \varphi^{g^-}(w_2)
     )
\sim_{\mathbb{P}_X} 
     \varphi^{g^- \circ f^-}(w_2) .
\end{equation*}

\qedhere
\end{enumerate}
\end{proof}

In Theorem \ref{thm:1},
two properties,
dynamic programming principle and time consistency are
usually introduced as axioms
(\cite{DS_2006}).
But, 
we derive them naturally here
from the fact that
the monetary value measure is a contravariant functor.

\lspace

Before ending this section,
we mention an interpretation of
the Yoneda lemma
in our setting.

\begin{thm} {[The Yoneda Lemma]}
\label{thm:yoneda}
For any monetary value measure
  $\Phi : \Prob^{op} \to \Set$
and an object
$
\ProbX
$
in
$\Prob$,
there exists a bijective correspondence 
$
  y_{\Phi, \ProbX}
$
specified by the following diagram:
\begin{equation*}
\xymatrix@C=30 pt@R=10 pt{
  y_{\Phi, \ProbX}
   :
  \Nat(\Prob(-, \ProbX), \Phi)
    \ar @{->}^-{
       \cong
    } [r]
&
  L^1(\ProbNX)
\\
  \alpha
    \ar @{|->} [r]
&
  \alpha_{\ProbX}
  (Id_X^-)
\\
  \tilde{\mathbf{u}}
&
  \mathbf{u}
    \ar @{|->} [l]
}
\end{equation*}
where
$\tilde{\mathbf{u}}$
is a natural transformation defined by
for any
$
f^-
  :
\ProbY
  \to
\ProbX
$
in
$\Prob$,
$
\tilde{\mathbf{u}}_{\ProbY}
  (f^-)
:=
\Phi f^- \mathbf{u}
$.
Moreover,
the correspondence is natural in both
$\Phi$
and
$\ProbX$.
\end{thm}

It makes sense to consider
the representable functor
$
  \Prob(-, \ProbX)
$
as a generalized 
\textit{time domain}
with time horizon
$
  \ProbX
$.
Then a natural transformation
from
$
  \Prob(-, \ProbX)
$
to
$
  \Phi
$
can be seen as a 
\textit{stochastic process}
that is (in a sense)
adapted to
$
  \Phi
$,
and its corresponding 
$\Sigma_X$-measurable
random variable
represents a terminal value (payoff) at the horizon.

The Yoneda lemma says that we have a bijective correspondence between 
those stochastic processes and random variables.


\def \imbed_ref  {1}

\if \imbed_ref  1

\else 

\bibliographystyle{apalike}
\bibliography{../../../taka_e}

\begin{thebibliography}{}

\bibitem[Adachi, 2014]{adachi_2014crm}
Adachi, T. (2014).
\newblock Toward categorical risk measure theory.
\newblock {\em Theory and Applications of Categories}, 29(14):389--405.

\bibitem[Adachi and Ryu, 2016]{AR_2016_1}
Adachi, T. and Ryu, Y. (2016).
\newblock A category of probability spaces.
\newblock https://arxiv.org/abs/1611.03630.

\bibitem[Detlefsen and Scandolo, 2006]{DS_2006}
Detlefsen, K. and Scandolo, G. (2006).
\newblock Conditional and dynamic convex risk measures.
\newblock Working paper.

\end{thebibliography}

\fi 

\end{document}